\documentclass[aps,preprint,groupedaddress]{revtex4-1}

\newcommand{\upperng}{\Omega_{\textrm{NG}}}
\newcommand{\upperu}{\Omega_{\textrm{U}}}
\newcommand{\fr}{\frac{\textrm{spikes}}{\textrm{second}}}

\usepackage{amsmath} 
\usepackage{amsthm} 

\theoremstyle{plain}
\newtheorem{theorem}{Theorem}[subsection]
\newtheorem*{theorem*}{Theorem}
\newtheorem{lemma}[theorem]{Lemma}
\newtheorem*{lemma*}{Lemma}

\theoremstyle{definition}

\newtheorem*{definition*}{Definition}

\newtheorem{example}[theorem]{Example}

\theoremstyle{remark}

\usepackage{graphicx}

\usepackage[margin=1.0in]{geometry}

\begin{document}

\title{Optimizing state change detection in functional temporal networks through dynamic community detection}

\author{Michael Vaiana}
\affiliation{Department of Mathematics and CDSE Program, University at Buffalo -
- SUNY, Buffalo, NY 14260}
\author{Sarah F. Muldoon}
\email{smuldoon@buffalo.edu}
\affiliation{Department of Mathematics and CDSE Program, University at Buffalo -
- SUNY, Buffalo, NY 14260}

\begin{abstract}
{Dynamic community detection provides a coherent description of network clusters over time, allowing one to track the growth and death of communities as the network evolves.  However, modularity maximization, a popular method for performing multilayer community detection, requires the specification of an appropriate null model as well as resolution and interlayer coupling parameters. Importantly, the ability of the algorithm to accurately detect community evolution is dependent on the choice of these parameters. In functional temporal networks, where evolving communities reflect changing functional relationships between network nodes, it is especially important that the detected communities reflect any state changes of the system.  Here, we present analytical work suggesting that a uniform null model provides improved sensitivity to the detection of small evolving communities in temporal correlation networks.  We then propose a method for increasing the sensitivity of modularity maximization to state changes in nodal dynamics by modeling self-identity links between layers based on the self-similarity of the network nodes between layers.  This method is more appropriate for functional temporal networks from both a modeling and mathematical perspective, as it incorporates the dynamic nature of network nodes.  We motivate our method based on applications in neuroscience where network nodes represent neurons and functional edges represent similarity of firing patterns in time.  Finally, we show that in simulated data sets of neuronal spike trains, updating interlayer links based on the firing properties of the neurons provides superior community detection of evolving network structure when group of neurons change their firing properties over time.}
{networks, modularity, community detection, multilayer, temporal}
\end{abstract}

\maketitle

\section{Introduction}
Networks are excellent data structures for capturing the pairwise interactions between a collection of objects, but when these interactions are dynamic, traditional network approaches require some form of data reduction to reduce the full range of interactions to a single summary.  As a result, the final network may or may not be a good representative of the underlying system.  A more suitable data structure for dynamic interactions is a \emph{temporal network} \cite{Holme2012,Kivela2014}, which is a time ordered collection of networks organized into layers together with \emph{interlayer edges} connecting nodes between layers. Temporal networks are increasingly being used to model time dependent properties of complex systems since no data loss is incurred as in the case of static networks \cite{bassett2015learningautonomy,valdano2015analytical,moinet2015burstiness,li2017fundamental,taylor2017eigenvector}. A particularly important class of temporal networks are functional networks
\cite{friston1994functional}, in which nodes represent a dynamic unit and edges are given by a measure of functional similarity (see Fig. \ref{fig:function_to_temporal}).  For example, in brain networks nodes may represent neurons and edges are measured via synchronization between temporal spiking activity patterns \cite{feldt2011dissecting}.

Since the nodes in a functional temporal network are themselves dynamic, the system may exhibit state changes where nodes switch between functional states.  Such a state change is depicted in Fig. \ref{fig:function_to_temporal} and occurs between time layers 2 and 3. These state changes effect the properties and structure of the network, and it is important to develop network measures that are maximally sensitive to such changes \cite{muldoon2018multilayer}.  Intuitively, one would expect interlayer edges to play a major role in temporal networks, as they model the strength of connections between nodes through time.  However, it is typically assumed that interlayer edges exist only between a node and itself in the next layer, representing a self-identity link through time, and more importantly, the standard assumption is that interlayer edge weights are all set to the same constant value through time \cite{Mucha2010,bassett2015learningautonomy,domenico2016mappingmultiplex, Bazzi2016}.  While the latter assumption has made temporal networks easier to deploy in practice (since tuning a single parameter for interlayer edge weights is significantly easier than tuning all possible edge weights individually), it is insufficient at properly modeling functional networks that exhibit state changes where the properties of a node change over time and the similarity of a node to itself in the previous layer is not always constant. 

One important property of a functional temporal network that is especially sensitive to state changes is the community structure \cite{fortunato2010community,mucha2010communities,vaiana2017multilayer}.  In temporal networks, communities can and do evolve through time, and nodal state changes can have drastic effects on this evolving community structure.  A popular method of community detection is modularity maximization \cite{Newman2004, Newman2006}, and this method has been extended to temporal networks \cite{Mucha2010}.  
While multilayer modularity maximization has been shown to provide a relatively efficient and effective description of dynamic community structure \cite{bassett2015learningautonomy,lee2016time,Yang2016,
weir2017post,muldoon2018locally} the multilayer modularity function has three parameters that must be specified: the resolution parameter, the interlayer coupling parameter, and the null model.  While some work has investigated how to choose these parameters in certain settings \cite{bassett2013robust,Bazzi2016,weir2017post}, how to optimally select these parameters remains an active area of research.  In this paper, we present work guiding the selection of these important parameters in order to make the function maximally sensitive to detecting state changes in dynamic community structure.  

First, expanding upon recent work that identifies an upper bound on the interlayer coupling parameter that limits the ability of multilayer modularity to detect certain changes in community structure across layers \cite{vaiana2018resolution}, we show that for functional temporal networks whose edge weights are bounded above by 1, using a uniform null model (as opposed to the commonly used Newman-Girvan null model) is preferable for detecting evolving communities.  We then propose a method to set interlayer edge weights in functional temporal networks based on a measure of temporal self-similarity.  In accordance with intuition, we lower the interlayer edge weight between a node and itself at the next time point when the functional similarity between that node and itself is low.  This gives a principled method of setting interlayer edge weights that further increases the sensitivity of modularity to detecting community changes.  Finally, we show how performing an essentially parameter free consensus, in combination with the methodology proposed above, gives nearly as accurate community detection results than those found with optimal choice of parameters when applied to simulated data representing neural activity with an implanted evolving community structure containing multiple state changes.

\begin{figure}[!ht]
\centering\includegraphics[width=4.5in]{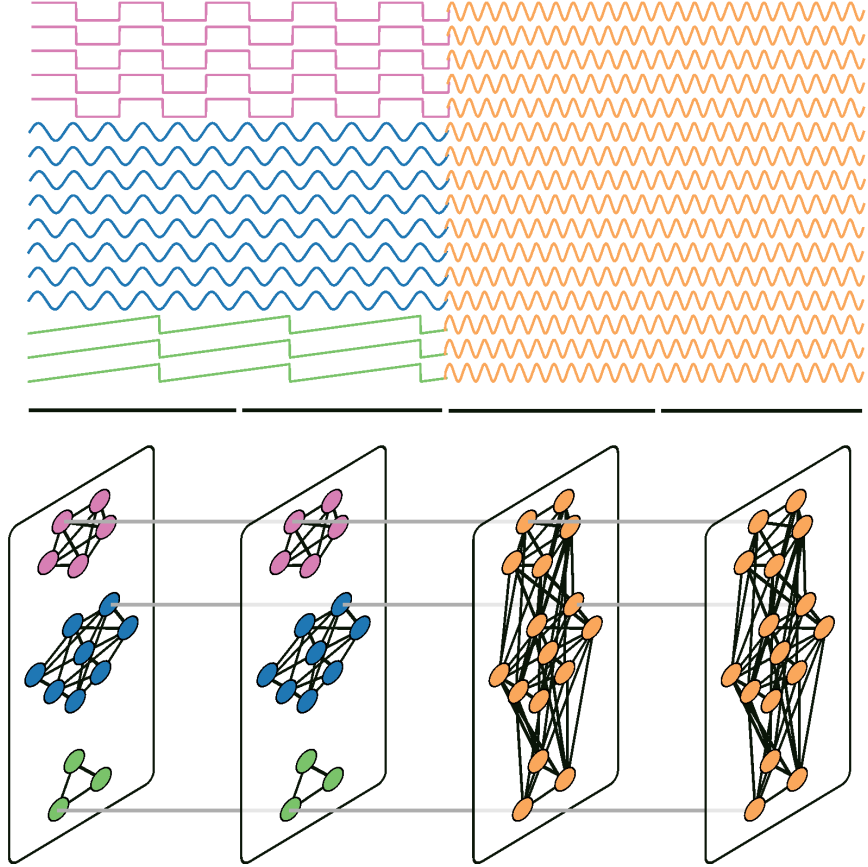}
\caption{\textbf{Functional Temporal Network.} The mapping of the activity of 16 functioning units to a temporal network.  Color indicates community assignment.  At the top there are 16 functioning nodes which initially are partitioned into 3 groups each with a unique function.  The system then undergoes a state change and the activity of the nodes synchronizes.  Black bars indicate windows in which a measure of similarity between function is measured.  The bottom is the temporal network with nodes representing the functioning units and edges representing the similarity in function in the given time window.  For visual clarity only 3 interlayer edges are drawn. Intuitively, the interlayer edge weight between layers 2 and 3 should be low to reflect the state change.}
\label{fig:function_to_temporal}
\end{figure}

The paper is organized as follows.  In Sec. \ref{sec:comm:review}, we briefly review temporal networks and the temporal modularity function.  In Sec. \ref{sec:parameter}, we discuss parameter selection for the modularity function and present our method of setting interlayer edge weights.  In Sec. \ref{sec:exp-results}, we test our method on simulated data and show that it gives improved sensitivity to the modularity function.

\section{Communities in Temporal networks}
\label{sec:comm:review}
\subsection{Temporal Networks and Notation}
A temporal network is a network $\mathcal{M} = (V,w,l)$ with $V$ the vertices, $w:V\times V \to R$ a weight function on the edges, and $l:V\to N\times T$ a labeling function.  The function $l$ identifies a vertex $v$ with a node-layer pair, that is $l(v) = (i,t)$ indicates that $v$ is the node $i$ at time $t$ and we denote this as vertex $i_t.$  The edges are implicitly defined via $w$ so that $w(i_s,j_t)$ is the weight of the edge between $i_s$ and $j_t.$ The edge between $i_s$ and $j_t$ is called an \emph{intralayer edge} if $s=t$ and an \emph{interlayer edge} if $s\neq t.$  We assume the network is diagonal and ordinal meaning there is no interlayer edge between $i_s, j_t$ unless $i=j$ and $s=t\pm 1.$ 
We let $A_{ijt}$ be the intralayer edge between node $i_t$ and $j_t$ so $A_{ijt}$ represents the adjacency matrix at time $t.$  Similarly, we let $B_{it}$ be the interlayer edge between node $i_t$ and $i_{t+1}.$

\subsection{Modularity in Temporal Networks}
Let $\mathcal{T}$ be a temporal network with $L$ many time layers. A community in a temporal network, $C$, is any subset of the vertices of the network, and importantly, communities can span multiple layers through time.  Let $P=\{C_1, C_2, \ldots C_k\}$ be a partition of the vertices of the network where each $C_i$ is a community.  The modularity function measures how well a given partition, $P$, captures the underlying community structure of the network.  It takes three parameters: an intralayer resolution parameter, $\gamma$, an interlayer coupling parameter, $\omega$, a null network, $R$, and is given by 
\begin{equation}
Q(P) = \sum_{t}^L \sum_{ij}(A_{ijt} - \gamma R_{ijt})\delta(c_{i_t}, c_{j_t}) + \sum_{t}^{L-1}\sum_i 2\omega\delta(c_{i_t}, c_{i_{t+1}})
\end{equation}
where $c_{i_t}$ is the community assignment of $i_t$ in $P$ and $\delta$ is the Kroneker delta function \cite{Mucha2010}.  For now, we have made the common assumption that $B_{it}=\omega$, where $\omega$ is a constant value across all nodes and layers, but we will relax this assumption later. The partition that maximizes the modularity function, $Q,$ consists of the putative communities.

\section{Optimal Modularity Parameters in Temporal Correlation Networks} 
\label{sec:parameter}
\subsection{Multilayer Resolution Limit and Null Models}
\label{sec:sub:mln-res-and-null-models}
The temporal modularity function is highly sensitive to its three parameters, the resolution parameter, $\gamma,$ the interlayer coupling parameter, $\omega,$ and the null network, $R.$  It is important to understand which choice of parameters is suitable for different types of networks.  In \cite{Bazzi2016}, an argument is made that in correlation networks, the uniform null network (U) may be more suitable than the classic Newman-Girvan null network (NG).  In \cite{vaiana2018resolution}, it was shown that all three parameters are related through a multilayer resolution limit: an upper bound on the interlayer coupling parameter which determines when modularity can detect mergers of communities. Below, we show how this resolution limit gives additional theoretical grounds for preferring the U null model in temporal networks whose edge weights are bounded above by 1, making the choice of the U null model particularly preferable in temporal networks whose edges are measured as correlations.

Let $\mathcal{T}$ be a temporal network with $L$ many time layers and $N$ many nodes on each layer. Let $K$ be a subset of the $N$ nodes and  assume that at time $t$ the nodes of $K$ form  communities $C_1, C_2,\ldots, C_r$ and that at time $t+1$ these all merge into a single community.  For example, the three communities in Fig. \ref{fig:function_to_temporal} merge together in layer 3 of the temporal network.

It was shown in \cite{vaiana2018resolution} that the merger of the communities can not be detected by modularity if $\omega > \Omega_t$, where 
\begin{equation}
\label{eq:Omega}
\Omega_t =  \sum_{\substack{ij\in K \\ \delta(c_{i_{t}}, c_{j_t}) = 0}} \frac{1}{2\theta}(\gamma R_{ijt} - A_{ijt})
\end{equation}
and $\theta = |K| - |C_m|$ where $C_m$ is the largest of the communities.   Therefore, $\Omega_t$ acts as a type of multilayer resolution limit on detecting changes in community structure across layers.  We refer to $\Omega_t$ throughout as the upper bound on interlayer coupling or simply the upper bound. Notice that all quantities in Eqn. \ref{eq:Omega} are with respect to layer $t.$  For notational convenience, we drop the subscript $t$ when it is understood we are working with a fixed layer.

Specializing to a null network lets one compute $\Omega_t$ more explicitly.  Assume we are working with a fixed layer $t$, and in this layer there are $n$ many communities $C_1,\ldots C_n$ that merge together at time $t+1.$   Define $\kappa_i$ = $\sum_{j}A_{ij}$ to be the degree of node $i$ and $2m = \sum_i \kappa_i$. The NG null network is then given by $R_{ij} = \frac{\kappa_i\kappa_j}{2m}.$   The U null network is given by $R_{ij} = \langle A \rangle$ for all $i,j$ where $\langle A \rangle$ is the mean value of $A.$  Explicitly, we have $\langle A \rangle = \frac{2m}{|N|^2}.$

In \cite{vaiana2018resolution}, it was shown that for the NG null network, the upper bound is given by 
\begin{equation}
\label{eq:ng-upper}
\upperng = \frac{1}{2\theta}\left(\frac{\gamma}{2m}\sum_{i\neq j}d_id_j - \sum_{i} e_i\right)
\end{equation} 
where $d_i$ is the degree of community $C_i$, and $e_i$ is the part of the external degree of $C_i$ that connects to another community $C_j.$  Recomputing $\Omega$ with respect to the U null network gives 
\begin{equation}
\label{eq:u-upper}
\upperu = \frac{1}{2\theta}\left(\gamma\sum_{i\neq j}|C_i||C_j| - \sum_{i} e_i\right)
\end{equation}
where $|C_i|$ is the number of nodes in community $C_i.$ Interestingly, for the U null network, the term that contributes positively to $\Omega$ does not depend on the edges within the communities and instead only depends on number of nodes within 
the communities.  

We now give conditions on the community structure that guarantee $\Omega_{U} > \Omega_{NG}.$ When this inequality holds, the upper bound for the $U$ null network is larger than that of the $NG$ network, thus implying there is a wider range of $\omega$ for which the communities $C_1, \ldots C_n$ can be detected.  We first prove a simple lemma that lets us approximate $\Omega_{NG}$ with $\hat{\Omega}_{NG}$ where $\Omega_{NG} \leq \hat{\Omega}_{NG}$.

\begin{lemma}
\label{thm:ng_bound}
Let $\mathcal{T}$ be a temporal network such that communities $C_1, \ldots C_n$ on layer $t$ merge together on layer $t+1.$  Let $K$ be the nodes of the communities $C_1 \ldots C_n.$  Then 
\[\Omega_{NG} \leq \frac{1}{2\theta}\left(\gamma 2k - \sum_i e_i\right) \equiv \hat{\Omega}_{NG}\] 
where $2k = \sum_{ij\in K} A_{ijt}.$ 
\end{lemma}
\begin{proof}
Since $K$ is a subset of nodes of the network we have $2k \leq 2m$, and thus 
\begin{align}
\label{eq:ng-approx}
\upperng &= \frac{1}{2\theta}\left(\frac{\gamma}{2m}\sum_{i\neq j}d_id_j - \sum_{i} e_i\right)
\\ &\leq \frac{1}{2\theta}\left(\frac{\gamma}{2m}(2k)^2 - \sum_{i} e_i\right)
\\ &\leq \frac{1}{2\theta}\left(\gamma 2k - \sum_{i} e_i \right)
\\ & \equiv \hat{\Omega}_{NG}.
\end{align}
\end{proof}

\begin{theorem}
\label{thm:u_bigger}
Let $\mathcal{T}$ be a temporal network such that communities $C_1, \ldots C_n$ on layer $t$ merge together on layer $t+1.$  Let $K$ be the nodes of the communities $C_1 \ldots C_n$ and let $\langle K \rangle$  be the average edge weight of those nodes.  Then, $\Omega_{U} > \Omega_{NG}$ if 
\[ \langle K \rangle < 1 - \frac{\sum_{i} |C_i|^2}{|K|^2} .\]
\end{theorem}

\begin{proof}
Our aim is to show that $\Omega_{U} > \Omega_{NG}.$  By Lemma \ref{thm:ng_bound}, it will suffice to show that $\Omega_{U} > \hat{\Omega}_{NG} > \Omega_{NG}.$  We compute 
\begin{align}
\Omega_{U} - \hat{\Omega}_{NG} &= \frac{1}{2\theta}\left(\gamma \sum_{i\neq j}|C_i||C_j| - \sum_{i}e_i\right) - \frac{1}{2\theta}\left(\gamma 2k - \sum_{i} e_i\right) \\
&= \frac{\gamma}{2\theta}\left(\sum_{i\neq j} |C_i||C_j| - 2k\right)\\
&= \frac{\gamma}{2\theta}\left(|K|^2 - \sum_{i}|C_i| - |K|^2\langle K\rangle \right).
\label{eq:explain1}
\end{align}
Where to obtain the last equality we used the fact that $2k = |K|^2\langle K \rangle$ and that 
\[|K|^2 = \sum_{i,j}|C_i||C_j| = \sum_{i\neq j} |C_i||C_j| + \sum_i|C_i|^2.\]  From the computation, we see that $\Omega_{U} > \hat{\Omega}_{NG}$ if and only if 
\[|K|^2 - \sum_{i}|C_i| - |K|^2\langle K\rangle  > 0 \]
and solving for $\langle K \rangle$ completes the proof.
\end{proof}

The content of the theorem says that if the average edge weight of the nodes that merge together is bounded by $ 1 - \frac{\sum_{i} |C_i|^2}{|K|^2}$, then we can conclude that $\Omega_{U} > \Omega_{NG}.$  In correlation networks, or any network whose edge weights are bounded by 1, the average edge value will also be bounded by 1.  Since $|K|^2 = \sum_{i\neq j} |C_i||C_j| + \sum_i|C_i|^2$, the quantity $\frac{\sum_{i} |C_i|^2}{|K|^2} \leq 1.$  How much less than one will depend upon the configuration of the communities that merge.  
\begin{example}
Assume the communities $C_1, \ldots, C_n$ from Theorem \ref{thm:u_bigger} all have the same size, $s.$  Then $\upperu > \upperng$ if $\langle K \rangle < 1 - \frac{1}{n}.$   We can see this by noting that $\sum_{i} |C_i|^2 = ns^2$ and $|K|^2 = (ns)^2$, and thus by the theorem, 
\[\langle K \rangle < 1 - \frac{ns^2}{(ns)^2} = 1 - \frac{1}{n}.\]
\end{example}
\noindent Therefore, in this example, if the average edge weight of the nodes that merge together is less than $1 - \frac{1}{n}$, it is preferable to use a uniform null network model.

There are a few final points that need to be made clear.  First, the converse of Theorem \ref{thm:u_bigger} does not hold, that is if $\langle K \rangle >1 - \frac{\sum_{i} |C_i|^2}{|K|^2}$ then nothing can be said about which bound, $\upperu$ or $\upperng$, is larger.  Second, the average edge weight, $\langle K \rangle$, is the average weight of the edges taken over all nodes in $K.$  Since we are assuming there are multiple communities in $K$, the connections between communities will necessarily be weak.  Thus, even in networks whose edge weights are bounded above by one, we expect the average value of the edge weights to be much lower for most community structures. It should also be noted that ensuring $\Omega_{U} > \Omega_{NG}$ does not itself mitigate the problem of the multilayer resolution limit.  However, it does provide a wider range of possible values of $\omega$ for which $\omega < \Omega$.

\subsection{Dynamic Interlayer Edges}

\label{sec:sub:dynamic-interlayer-edges}
Community detection in functional temporal networks presents several challenges.  We would like to model the underlying system in a way that is maximally sensitive to changes in state or behavior.  The first step in improving such sensitivity is to leverage the interlayer edges by letting them reflect a measure of self-similarity in nodes instead of being held constant (i.e., we will no longer assume $B_{it}=\omega$).  Such a measure of self-similarity is often available in functional networks, and we assume that the self-similarity of node, $i_t$, to itself one time step later, $i_{t+1}$, is given by $s(i_t, i_{t+1}).$ Practically, one would define this measure based on some sort of nodal feature that quantifies nodal dynamics and that that differs in different states of the system.

Our method is to first choose $\omega_{\textrm{global}}$ and set $B_{it}=\omega_{\textrm{global}}$ for all $i$ and $t$ (recall $B_{it}$ gives the interlayer edge between $i_t$ and $i_{t+1}).$   We then update $B_{it}$ on a node-by-node, layer-by-layer basis, depending on the value of $s(i_t, i_{t+1}).$  If we let $\rho \leq 1$ be a percentage, then our method is summarized as 
\begin{enumerate}
\item Choose $\omega_{\textrm{global}}$ and set $B_{it} =\omega_{\textrm{global}}$ for all $i$ and $t$. 
\item For each $i$ and $t$ if $s(i_t, i_{t+1})$ is small, then set $B_{it} = \rho\omega_{\textrm{global}}$.
\end{enumerate}

We refer to this process as \emph{updating interlayer edges}. The meaning of $s(i_t,i_{t+1})$ being too small is application dependent.  In our computational application in Section \ref{sec:exp-results}, we choose a threshold $\tau$ and say that if $s(i_t, i_{t+1}) < \tau$, then we will update the interlayer edges. However, any function of the self-similarity value can be used to determine when to update edges.

\subsection{Parameter Consensus}
\label{sec:sub:practical}
Maximizing the modularity function is computationally infeasible so in practice one uses a heuristic, for example the louvain algorithm \cite{Blondel2008,Jutla2011}. It has been shown that the modularity landscape has a large number of locally optimal partitions \cite{Good2010} which can correspond to significantly different community assignments.  To combat this dilemma, one can run a heuristic many times and form a consensus partition over the different runs \cite{bassett2013robust,lancichinetti2012consensus}. The idea is that the true core structure should be detected by the algorithm over most of the runs while weaker structure may only get detected a few times.  Using a consensus can then determine the structure that is consistently found.

There are several different approaches to applying a consensus algorithm \cite{bassett2013robust,Sarzynska2014,lancichinetti2012consensus,
jeub2018multiresolution,seifi2013stable,weir2017post}, but here we focus on methods that allow us to incorporate our method of updating interlayer edges. Because the algorithm for maximizing modularity is stochastic, one may choose a fixed pair of parameters $(\gamma_0, \omega_0)$ and run the algorithm many times with this fixed choice.  Alternatively, one can chose a set of pairs $\{(\gamma_0, \omega_0), (\gamma_1, \omega_1), \cdots (\gamma_n, \omega_n)\}$ and run the algorithm once for each of these pairs.  In this formulation, the resulting structure one obtains will be consistently found over different spatial resolutions (due to varying $\gamma$) and temporal resolutions (due to varying $\omega$).

To understand the effects of our method of updating interlayer edges, we run community detection in four different ways.  First, we run a fixed consensus, i.e., is we choose a fixed $(\gamma, \omega)$, and we run a consensus on these parameters without updating interlayer edges. We refer to this a \emph{fixed consensus}.  Next, we choose a grid of points in the $\gamma, \omega$ plane and run a consensus over these parameters.  We refer to this as a \emph{sweep consensus} since we are sweeping the parameter space over some grid.  We then run both the fixed and sweep consensus but in addition we update interlayer edges according to our method in Section \ref{sec:sub:dynamic-interlayer-edges}, and we refer to this as \emph{fixed with updates} and \emph{sweep with updates} respectively.

\section{Optimal parameter selection in dynamic neuronal data}
\label{sec:exp-results}
We now apply our proposed methodology to simulated neural spike train data with multiple embedded state changes to show that the use of a uniform null network model and dynamically updating interlayer edges gives superior performance to traditional methods for detecting evolving community structure. We stochastically generate 100 synthetic temporal networks with the same planted community structure and attempt to recover these communities via multilayer modularity maximization. Since our networks are created synthetically, we have access to the true community structure.  To compare the results of a community consensus with the true community structure we measure the Normalized Mutual Information (NMI) between the two partitions \cite{Danon2005}. The NMI is bounded between 0 and 1 with where two partitions which match perfectly will have NMI equal to 1 and partitions which disagree tend to 0.

\subsection{Network Creation}
\label{sec:net-creation}
The networks are created by measuring correlations between the activity of synthetic neurons (Fig. \ref{fig:spike_networks}).  The neural activity is generated with two goals in mind.  First, we need to create a structure with state changes in order to test our method of updating interlayer edges.  We are also motivated by brain activity  
in which neurons are sequentially recruited into a highly active and synchronized state, such as might occur in seizure dynamics.  We would therefore like our network to model this type of spreading activity. To accomplish these goals, we generate a network in which a single community of highly active neurons sequentially merges with other low activity neurons.

\begin{figure}[!ht]
	\begin{center}
	\includegraphics[width=.6\textwidth]{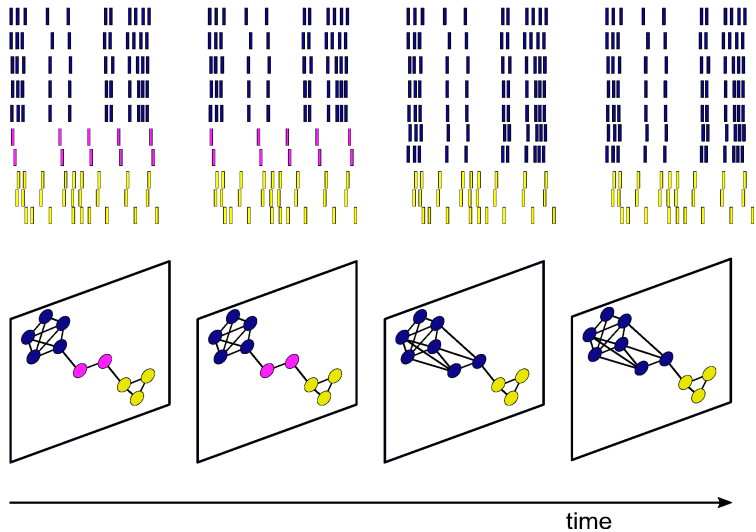}
	\caption{\textbf{Mapping Spike Trains to Temporal Networks}.  \textit{Top}.  Spike train data for 10 neurons.  The rows correspond to the activity of the neurons, and a bar represents a time point at which the neuron was active or `spiked'.  \textit{Bottom}.  A temporal network created by measuring correlations in activity between spikes of neurons. Color indicates community assignments based on similarity in firing activity.}
	\label{fig:spike_networks}
	\end{center}
\end{figure}

Specifically, we generate 12 seconds of activity for 100 neurons via a Poisson process \cite{Feldt2009}.  The Poisson process requires a firing rate parameter which controls the number of spikes per second a neuron will fire on average.  The firing rate parameter for each neuron is initially set to 10 $\fr.$  The network is created in such a way that each neuron undergoes a state change at which point the firing rate increases to 30 $\fr$ and the neuron synchronizes its firing pattern with all other neurons with this increased firing rate.  In the first 1-2s of activity, all neurons are generated with a firing rate of 10 $\fr.$  In seconds 2-3, 10 of the 100 neurons undergo the state change (their firing rate parameter is increased and their activity becomes synchronized).  This synchronization is achieved by generating a single master spike train and then using this train to generate correlated trains.  The correlated trains are created by randomly jittering the positions of spikes (within a 5ms window) and randomly deleting 10\% of the spikes.  In time 3-4s, another 10 neurons undergo a state change and all synchronize with the previous 10 synchronized neurons, forming a group of 20 neurons with an increased firing rate and synchronous activity.  This process continues until seconds 10-12 where all neurons are synchronized and remain this way. We build a functional network by measuring correlations between the neural activity of each neuron in a given time window. The choice of time window will affect the correlations measured within that time window and can have an impact on the results of community detection \cite{Telseford2016}.  To account for this, we choose three time windows of size 1, 1.5, and 2 seconds, respectively.  These choices represent different levels of mismatch between the window and the underlying structure.  We call the 1 second window `matching', the 1.5 second window `disjoint' and the 2 second window `large' to indicate the level at which the window overlaps the dynamic structure of changes within the network.  See Fig. \ref{fig:intraedges} for an example of the intralayer adjacency matrices for the matching window case.

\begin{figure}[!ht]
	\begin{center}
	\includegraphics[width=\textwidth]{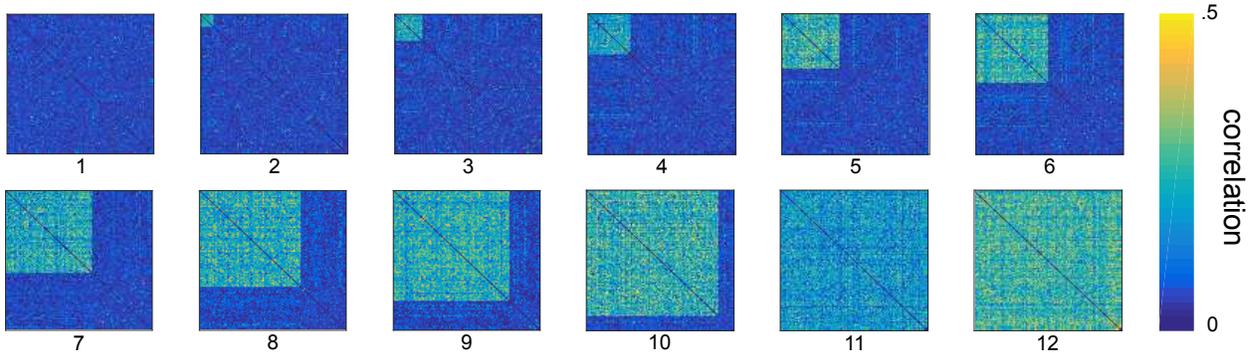}
	\caption{\textbf{Temporal Network Structure}. The adjacency matrices for each layer of a sample network with matching window size (1s). The $i,j$ value of each matrix is the correlation between the activity of neuron $i$ and neuron $j$ in the given time window. Notice that at each layer, 10 neurons join the growing synchronized community.  This network structure represents a sequence of mergers of communities.}
	\label{fig:intraedges}
	\end{center}
\end{figure}

\subsection{Parameter Sweep and Null Networks}
To assess the performance of our method on the simulated network, we attempt to find near optimal parameters ($\gamma$ and $\omega$) with which to compare our results.  To do this, we choose $(\gamma,\omega)$ in $[0,2]\times [0,2]$ discretized by a step size of .05.  This gives 1681 points in in the parameter plane, and for each of these points, we run a fixed consensus over 100 runs of the modularity maximization algorithm.  We do this for each of the 100 sample networks so that, for each sample, we can determine the parameters that give the highest NMI when using a fixed consensus (optimal parameters).  The parameters found this way are optimal in the sense that they maximize the NMI with respect to the ground truth, and thus represent the best choice of parameters. 

Note that in practice, doing such a sweep to find optimal parameters is not possible since one can not compare the output partition to the ground truth.  In addition, performing a fixed consensus (100 runs of the heuristic) over 1681 possible pairs of parameters is computationally expensive.  Thus, while we perform this task here for comparison, performing a parameter sweep to find the optimal parameters is not feasible option.

In Fig. \ref{fig:sweep_results}, we show the mean NMI over all 100 samples for each choice of parameter value and for each of the two null networks in consideration.  Notice that when using the $NG$ null network, there are two phases of the parameter plane separated by $\gamma \sim.9$ and that, in general, the NMI for the $NG$ network is lower than that of the $U$ network.

Although the exact values of optimal parameters varies  within the 100 sample networks, as seen in Fig \ref{fig:sweep_results}, these values are approximately within the range $\omega \in [.3,1]$ and $\gamma\in [.9,1.3]$ for the $U$ null network and $\omega \in (0,2], \gamma\in [0,.9]$ for the $NG$ network.

\begin{figure}[!ht]
\centering
\includegraphics[width=.8\textwidth]{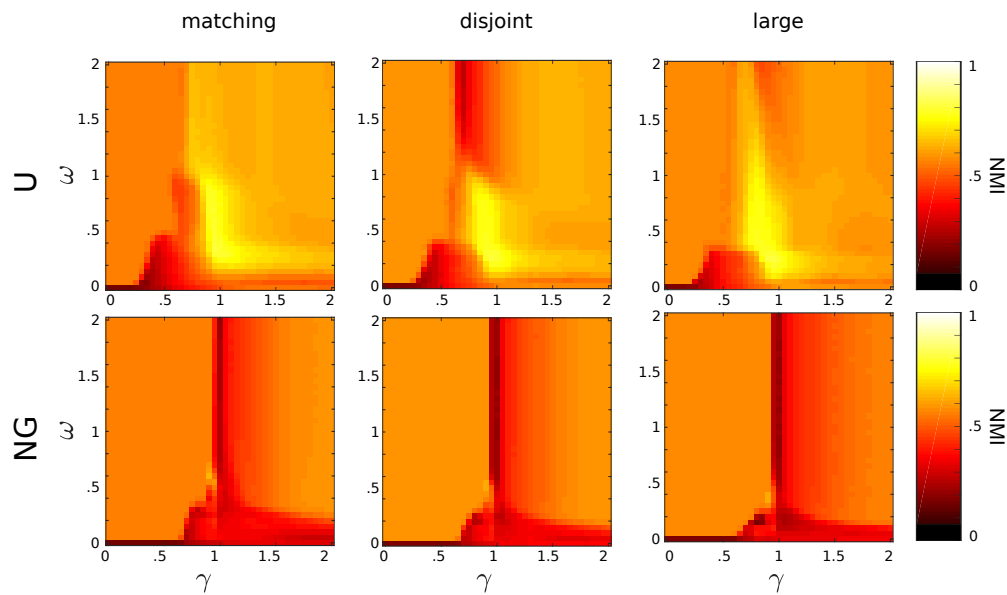}
\caption{\textbf{Modularity Parameter Landscape}. The mean NMI over all 100 samples for each choice of parameters. \emph{Top}. The mean NMI using the uniform null network for each window choice and for $\omega\in [0,2]$ and $\gamma\in [0,2].$ \emph{Bottom.} The mean NMI using the NG null network for each window choice and for $\omega\in [0,2]$ and $\gamma\in [0,2].$  Notice that the uniform null network has a region of high NMI for $.2 < \omega < 1$ and $\gamma \sim 1$ whereas the NG null network exhibits a sharp phase transition in when $\gamma \sim 1.$ }
\label{fig:sweep_results}
\end{figure}

\subsection{Comparing Community Detection Parameters}
\label{sec:sub:results}
Having found the optimal parameters for each of the 100 sample networks over all 3 time windows and each of the 2 null networks (U, NG), we can compare our method with the performance of the algorithm with these optimal parameters.  We run a fixed consensus with optimal parameters both with and without updates, and we also run a sweep consensus with and without updates.  

To perform updates, we first measure the firing rate for each neuron in each time window.  For a fixed pair of consecutive time windows, $t_1$ and $t_2$, for each neuron, we compute the difference in firing rate between these windows.  If the difference in firing rates is 2 or more standard deviations above or below the mean, we update the interlayer edge value between windows $t_1$ and $t_2$ according to the procedure in Section \ref{sec:sub:dynamic-interlayer-edges}.  Thus, our notion of similarity is based on population statistics: if a neuron's change in firing rate is significant relative to the population, we 
update its interlayer edge to reflect the decreased self-similarity between time layers.

In Fig. \ref{fig:violins}, we show the distribution of NMI for each of the 300 networks (100 samples and 3 time windows) and for both null networks.  It is immediately clear that the U null network outperforms the NG network on all fronts.  Notice that for the U null network, the sweep consensus with updates performs comparably to the fixed consensus with and without updates. Intuitively, one would expect the optimal parameters to significantly outperform the sweep consensus since the comparison is between parameters found to be optimal with an indiscriminate range of parameters.  Surprisingly, the sweep consensus together with updating interlayer edge weights gives similar performance to that with optimal parameters, indicating it may be a promising tool in practice when optimal parameters are not available.

\begin{figure}[!ht]
	\begin{center}
	\includegraphics[width=\textwidth]{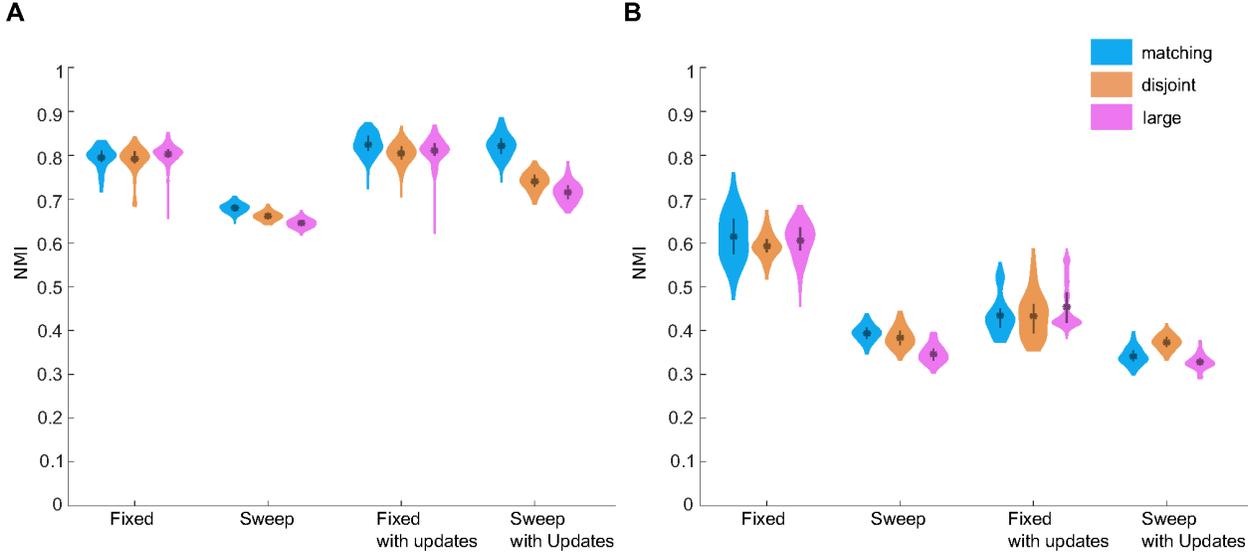}
	\caption{\textbf{NMI Distributions}. The normalized mutual information (NMI) over the 100 samples for each time window and each method of parameter selection. \textbf{A)} Results for the U null network. Notice that a sweep with updates performs nearly as well as a consensus over an optimal choice of parameters.  \textbf{B)} Results for the NG null network.  The NMI for the NG null network is fairly low across all methods.}
	\label{fig:violins}
	\end{center}
\end{figure}

To further quantify the performance of these methods, we show the dynamic community structure found via each method in Figs. \ref{fig:all_results_u} and \ref{fig:all_results_c}.  Since we have 100 sample networks, we only choose a single representative result which is chosen so that the NMI of the representative is close to the average NMI over all 100 samples. That is, the displayed community structure represents the average community detection performance.  We color the larger growing community in pink for better color differentiation.

\begin{figure}[!ht]
\centering
\includegraphics[width=.8\textwidth]{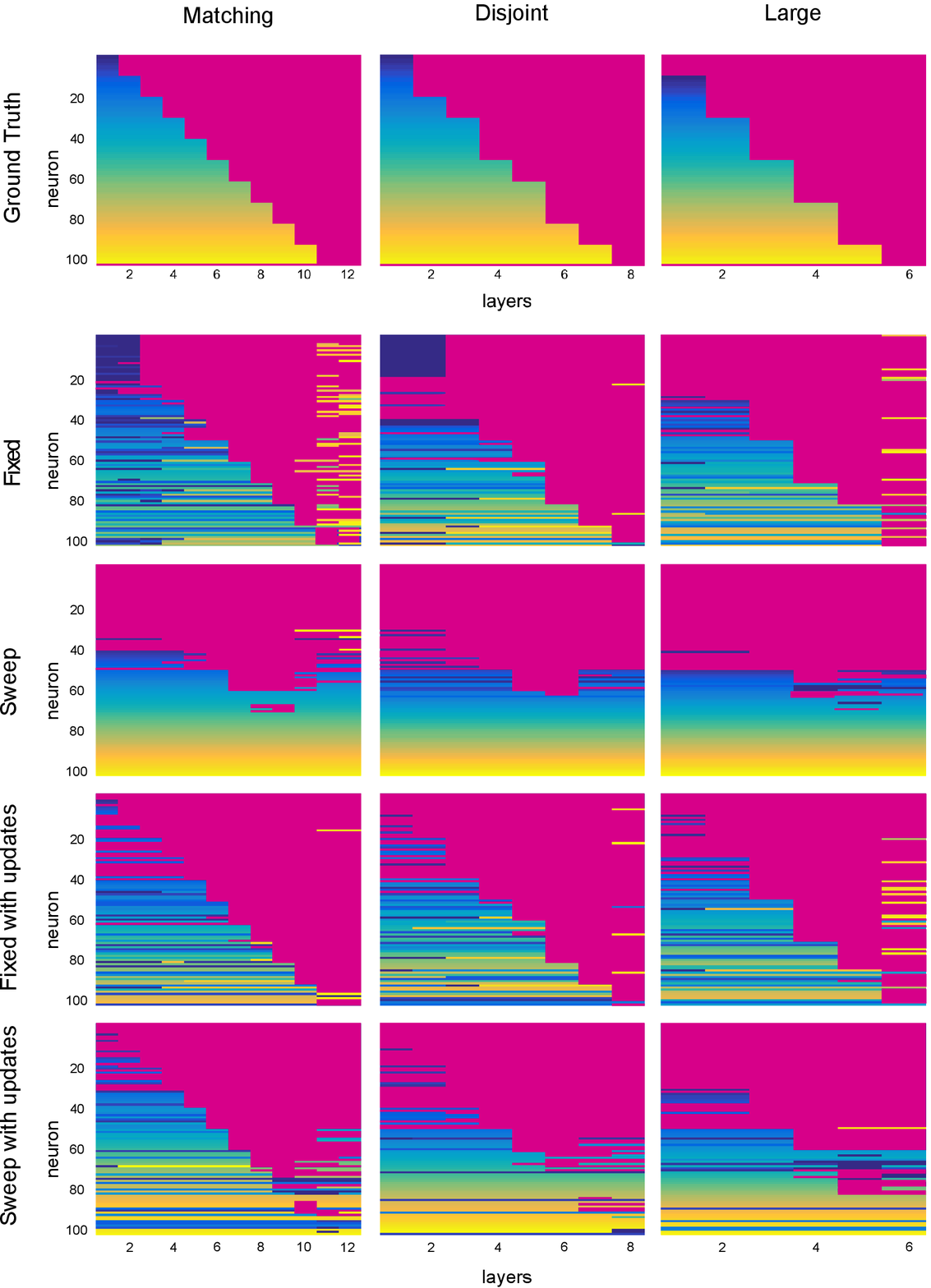}
\caption{\textbf{Average Community Detection for the U Null Network}. Representative partitions found by the the U null network for each method and time window.  The representative is chosen so that the NMI of the representative with the ground truth is close to the average NMI taken over all 100 samples. }
	\label{fig:all_results_u}
\end{figure}
\begin{figure}[!ht]
\centering
\includegraphics[width=.8\textwidth]{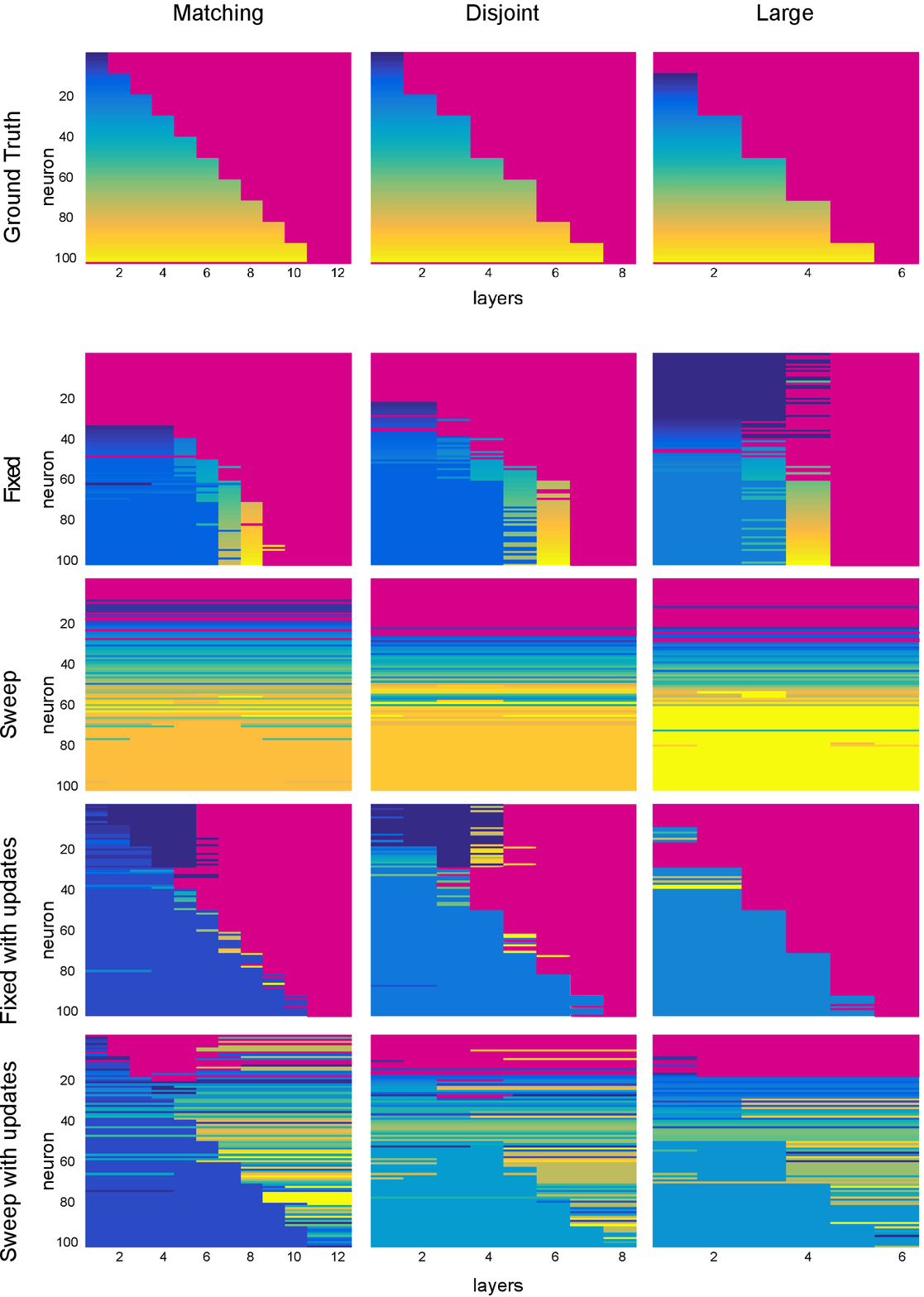}
\caption{\textbf{Average Community Detection for the NG Null Network}. Representative partitions found by the the NG null network for each method and time window.The representative is chosen so that the NMI of the representative with the ground truth is close to the average NMI taken over all 100 samples.}
	\label{fig:all_results_c}
\end{figure}

It should be made clear that the NMI, like all measures of similarity between partitions, is not perfect.  Although the NMI distributions are drastically lower for the NG null network, Figures \ref{fig:all_results_u} and \ref{fig:all_results_c} provide indications of what drives the differences. The NG null network tends to group the uncorrelated nodes into one large community where as the U null network does a better job of separating them.  This is likely driving the majority of the difference in the NMI between the two choices of null models.  For the fixed method with the NG network, one can clearly see the multilayer resolution limit take effect in the first and last time layers, where small community mergers can not be detected.  With the U null network, the resolution limit is not as drastic and is not present in the last layers of the network.  This is what we expect from the formula for $\upperu$ given in Eqn. \ref{eq:u-upper}.  Finally, we notice that for the NG null network and the sweep with updates method, the uncorrelated nodes are grouped into a single community while the larger community of correlated nodes are separated.  This is the exact opposite of the ground truth.  This phenonmenon is not present in the U null network, which gives very good performance with the sweep consensus with updates.  

 In the current study, we were only able to measure optimal parameters because we had access to the ground truth.  However, our main result suggests that using the U null network with a sweep consensus and updates gives nearly optimal performance. In most real-world settings, the ground truth is not known and one will not have access to the optimal parameters, but one may have access to a self-similarity measure for the nodes.  In this case, a sweep consensus together with interlayer updates would be appropriate.   
\clearpage

\section{Discussion and Conclusions}
Functional temporal networks are an import class of networks for modeling dynamic nodes, and the community structure in these networks can evolve through the temporal layers.  It is therefore important to develop methods that are sensitive to these community changes, especially when the system undergoes a state change.  In this paper, we addressed this issue with regards to the popular modularity maximization method for dynamic community detection. Our contribution is three-fold: we (i) gave a theoretical result on the difference in choice of two popular null modals for modularity maximization in functional temporal networks; (ii) introduced a method for setting interlayer coupling on a layer by layer basis; and (iii) showed how a minimal parameter consensus together with our method provides a robust method of community detection when applied to simulated neural data that exhibits state changes.

Although when using multilayer networks to model physical systems it is commonly assumed that interlayer coupling is constant, this leads to an inflexible model that does not necessarily reflect the dynamics of the system. Here, we have proposed a method that updates interlayer coupling values based on properties of nodal self-similarity between layers.  Our observation that the representation of interlayer links should be reflective of system states over time could of course be extended further, and all values of interlayer coupling could be derived based on experimental measures of nodal similarity over time.  We hope that the findings presented here will drive future work exploring more complicated and realistic methods for setting interlayer coupling values.  However, our results do suggest that even our simple method of updating interlayer coupling values provides increased sensitivity to the detection of state changes in functional temporal networks, and we encourage others to explore the use of this method in a wider array of dynamic networks where state changes are observed.

\begin{acknowledgments}
The authors would like to acknowledge support from the National Science 
Foundation (SMA-1734795).  The content is solely the responsibility of the autho
rs and does not necessarily represent the official views of the funding agency.
\end{acknowledgments}



%


%

\end{document}